\documentclass[12pt,a4paper]{amsart}
\usepackage{amsfonts,amssymb,latexsym,amsmath, amsxtra, eufrak, mathrsfs}
\usepackage[all]{xy}
\usepackage{graphics,graphicx,epsfig}

\usepackage[OT2,T1]{fontenc}
\DeclareSymbolFont{cyrletters}{OT2}{wncyr}{m}{n}
\DeclareMathSymbol{\Sha}{\mathalpha}{cyrletters}{"58}

\theoremstyle{plain}
\newtheorem{theorem}{Theorem}[section]

\newtheorem{conjecture}[theorem]{Conjecture}

\newtheorem{proposition}[theorem]{Proposition}

\newtheorem*{conjecture*}{Conjecture}

\newtheorem{definition}[theorem]{Definition}

\numberwithin{equation}{section}

\let\non\nonumber

\newcommand{\bea}{\begin{eqnarray}}
\newcommand{\eea}{\end{eqnarray}}
\newcommand{\be}{\begin{equation}}
\newcommand{\ee}{\end{equation}}

\newcommand{\sgn}{\mathrm{sgn}}
\newcommand{\Ext}{\mathrm{Ext}}

\newcommand{\noi}{\noindent}

\newcommand{\Aut}{\mathrm{Aut}}

\newcommand{\bF}{\mathbb{F}}
\newcommand{\bP}{\mathbb{P}}

\newcommand{\half}{\textstyle{\frac{1}{2}}}

\newcommand{\CH}{\mathcal{H}}
\newcommand{\CI}{\mathcal{I}}

\newcommand{\CL}{\mathcal{L}}
\newcommand{\CM}{\mathcal{M}}
\newcommand{\CN}{\mathcal{N}}
\newcommand{\CO}{\mathcal{O}}


 
\begin{document}

\title[BPS invariants of semi-stable sheaves on rational surfaces]{BPS invariants of semi-stable sheaves on
rational surfaces}  

\author{Jan Manschot}
\address{Max Planck Institute for Mathematics, Vivatsgasse 7, 53111 Bonn, Germany}
\address{Bethe Center for Theoretical Physics, Bonn University, Nu\ss allee 12, 53115 Bonn, Germany}
\email{manschot@uni-bonn.de}

\begin{abstract}
\baselineskip=18pt 
\noi BPS invariants are computed, capturing topological invariants of moduli spaces of
semi-stable sheaves on rational surfaces. For a suitable stability
condition, it is proposed that the generating function of BPS invariants of a Hirzebruch surface
$\Sigma_\ell$ takes the form of a product formula. BPS invariants for
other stability conditions and other rational surfaces are obtained using 
Harder-Narasimhan filtrations and the blow-up
formula. Explicit expressions are given for rank $\leq 3$ sheaves on  
$\Sigma_\ell$ or the projective plane $\bP^2$. The applied techniques
can be applied iteratively to compute invariants for higher rank.
\end{abstract}
\maketitle

\baselineskip=19pt 


\section{Introduction}

Topological invariants of moduli spaces of semi-stable sheaves on complex surfaces are a rich subject with links to many topics 
in physics and mathematics. Closely related topics in physics are gauge theory, instantons, electric-magnetic
duality \cite{Vafa:1994tf} and also (multi-center) black holes
\cite{Diaconescu:2007bf, Manschot:2009ia, Manschot:2010xp}. Instantons
saturate the bound on their minimal action, the so-called
Bogomolnyi-Prasad-Sommerfeld (BPS) bound. The prime interest of
this article are topological invariants of moduli spaces of instantons, in particular their Poincar\'e
polynomials, which are commonly referred to as ``BPS invariants''. These invariants
correspond also to (refined) supersymmetric indices enumerating
supersymmetric or BPS states.

Instantons on complex surfaces are described algebraically as semi-stable vector
bundles and coherent sheaves \cite{Huybrechts:1996,
  Friedman:1998}. Generating functions of BPS invariants of sheaves on
surfaces are computed for rank 1 by G\"ottsche \cite{Gottsche:1990} 
and rank 2 by Yoshioka \cite{Yoshioka:1994, Yoshioka:1995}. These
generating functions lead to intriguing connections with (mock) modular forms 
\cite{Vafa:1994tf,Gottsche:1996, Gottsche:1998,Bringmann:2010sd},
which are a manifestation of electric-magnetic duality of the gauge theory \cite{Vafa:1994tf}.
Refs. \cite{Manschot:2010nc, Manschot:2011dj} compute BPS
invariants for rank 3 sheaves with Chern classes such that stability coincides
with semi-stability. 

The present article computes the BPS invariants of
{\it semi-stable} sheaves with rank 3 on Hirzebruch surfaces $\Sigma_\ell$ and on the projective plane $\bP^2$, and explains how to generalize the computations to 
higher rank. The developed techniques can be applied straightforwardly to
compute BPS invariants of the other rational and ruled surfaces. Although the extension from stable to semi-stable might seem a minor one, it requires to deal
with various subtle but fundamental aspects of the moduli spaces of semi-stable
sheaves, which could be neglected in Ref. \cite{Manschot:2010nc}. 
Having resolved how to deal with these aspects for $r=3$, the computations can
in principle be extended to any rank.
 
This introduction continues with summarizing the contents of the paper, after recalling the computations in
Ref. \cite{Manschot:2010nc} which were inspired by 
\cite{ Yoshioka:1994, Yoshioka:1995, Gottsche:1996,
  Gottsche:1998}. A crucial fact for the computations is that the blow-up $\phi:\tilde
\bP^2\to \bP^2$ is isomorphic to the Hirzebruch surface $\Sigma_{\ell}\to
C$ with $\ell=1$. The fibre $f$ and base $C$ of $\Sigma_{\ell}$ are both isomorphic to $\bP^1$. 
As explained in more detail in Section \ref{subsec:suitable}, the BPS invariants of $\Sigma_\ell$
with polarization $J$ chosen sufficiently close to $f$ (a so-called ``suitable''
polarization, see Definition \ref{def:suitable}) vanish for
sheaves with first Chern class $c_1$ and rank $r$ such that $c_1\cdot f\neq 0 \mod r$.

Wall-crossing then allowed to compute the BPS invariants for other choices of $J$. The
BPS invariants of $\bP^2$ were obtained from those of $\tilde \bP^2$ by application of the blow-up
formula \cite{Yoshioka:1996, Gottsche:1998, Li:1999}, which is a simple relation between the
generating functions of the invariants for $\bP^2$ and $\tilde
\bP^2$. However, its original form
is only valid for $\gcd(c_1\cdot \phi^*H,r)=1$ and
$J=\phi^* H$, with $H$ the hyperplane class of $\bP^2$. 

The present paper describes how to deal with the cases when $c_1$ and $r$ do not satisfy the 
constraints for vanishing of the BPS invariant or the blow-up
formula. The formal theory of invariants of moduli spaces (or stacks) of
semi-stable sheaves is developed by Kontsevich and Soibelman
\cite{Kontsevich:2008} and Joyce \cite{Joyce:2004, Joyce:2005, Joyce:2008}. We will 
in particular use the notion of virtual Poincar\'e functions for
moduli stacks, which are a generalization of Poincar\'e polynomials of
manifolds. The virtual Poincar\'e function of a moduli stack is
(conjecturally) related to the BPS invariant by
(\ref{eq:stackinvariant}). The BPS invariant is most natural from
physics and leads to generating functions with modular properties.  

Two novel ingredients of this paper are:
\begin{enumerate}
\item Eq. (\ref{eq:restrictfibre}) which provides for any rank
  $r\geq 1$ the generating function of virtual Poincar\'e functions of the
  moduli stack of sheaves on a Hirzebruch surface $\Sigma_\ell$ whose 
restriction to the fibre $f$ is semi-stable. Eq. (\ref{eq:totalset2})
gives the generalization to virtual Hodge functions for more general ruled surfaces $\Sigma_{g,\ell}$.
\item {\it Extended} Harder-Narasimhan filtrations $0\subset F_1\subset F_2\subset\dots  \subset
  F_\ell=F$, whose definition (Def. \ref{def:extHNfiltr}) differs from
  the usual definition (\ref{def:HNfiltr}) of HN filtrations by
  allowing quotients $E_i=F_i/F_{i-1}$ with equal (Gieseker)
  stability $p_J(E_i,n)\succeq p_J(E_{i+1},n)$. These filtrations in
  combination with the associated invariants (\ref{eq:setfiltration}) are
  particularly useful to compute generating functions of BPS
  invariants starting from Conjecture \ref{eq:restrictfibre} and their
  changes across walls of marginal stability. 
\end{enumerate}


To obtain the BPS invariants for a suitable polarization, one
subtracts from Eq. (\ref{eq:restrictfibre}) generating functions
corresponding to extended HN-filtrations given by (\ref{eq:setfiltration}), analogous to the seminal
papers about vector bundles on
curves \cite{Harder:1975, Atiyah:1982fa}.  Naturally, these techniques are also applicable to compute invariants
of semi-stable invariants for other mathematical objects like vector bundles on curves
and quivers. Also a solution to this
recursive procedure is given analogous to Ref. \cite{Zagier:1996}. Then repeated application
of the formula for filtrations (which is equivalent with the
wall-crossing formulas \cite{Kontsevich:2008, 
    Joyce:2008}) gives the BPS invariants for other choices of the
  polarization.

Finally, the blow-up formula provides the invariants on $\bP^2$. The
earlier mentioned condition $\gcd(c_1\cdot \phi^*H,r)=1$ is a
consequence of the fact that the blow-up formula is applicable for the
Poincar\'e functions $\CI^\mu(\Gamma,w;J)$ with respect to
$\mu$-stability instead of the more refined Gieseker
stability. However with the invariant for filtrations (\ref{eq:setfiltration}), it is
straightforward to transform the BPS invariants $\Omega(\Gamma,w;J)$
to $\CI^\mu(\Gamma,w;J)$ for $\mu$-stability. The rational factors in Eq.
(\ref{eq:setfiltration}) appear naturally in the relation between the
generating functions of these invariants.

The paper illustrates in detail the above steps for sheaves with rank
2 and 3, and shows their agreement with various consistency
conditions, e.g. the blow-up formula,
integrality and  $w\leftrightarrow w^{-1}$ symmetry of the Poincar\'e
polynomial. 
\\
\newline 
The outline of the paper is as follows. Section \ref{sec:sheaves}
reviews some necessary properties of sheaves on surfaces including
stability conditions. Section \ref{sec:genfunctions} discusses the
invariants and generating functions. Section \ref{sec:restrictfibre}
presents the generating function (\ref{eq:restrictfibre}) of the
virtual Poincar\'e functions of the stack of sheaves whose restriction
to the fibre is semi-stable. Then
we continue with the computation of the invariants of $\Sigma_\ell$
for any choice of polarization in Section
\ref{sec:ruledsurface}. Finally Section \ref{sec:projplane} presents the blow-up
formula (\ref{eq:blowup}) and computes
the generating function for sheaves on $\bP^2$ with $(r,c_1)=(3,0)$. 

\section*{Acknowledgements}
\noi I would like to thank L. G\"ottsche, H. Nakajima, T. Wotschke and
K. Yoshioka for helpful and inspiring discussions. I am grateful to
E. Diaconescu and especially S. Meinhardt for their explanations of
the work of D. Joyce \cite{Joyce:2004, Joyce:2005}. Part of the presented research was done as a postdoc of the IPhT,
CEA Saclay and supported by ANR grant  BLAN06-3-137168.

\section{Sheaves on surfaces}
\label{sec:sheaves} 

We consider sheaves on a smooth projective surface $S$. The Chern character of the sheaf $F$ is given by
ch$(F)=r(F)+c_1(F)+\frac{1}{2}c_1(F)^2-c_2(F)$ in terms of the
rank  $r(F)$ and its Chern classes $c_1(F)$ and $c_2(F)$. The vector
$\Gamma(F)$ parametrizes in the following the topological classes of
the sheaf $\Gamma(F):=(\,r(F),\mathrm{ch}_1(F),\mathrm{ch}_2(F)\,)$. Other frequently
occuring quantities are the determinant 
$\Delta(F)=\frac{1}{r(F)}(c_2(F)-\frac{r(F)-1}{2r(F)}c_1(F)^2)$, and
$\mu(F)=c_1(F)/r(F)\in H^2(S,\mathbb{Q})$. 

Given a filtration $0\subset F_1\subset \dots \subset F_{\ell}=F$, let
$E_i=F_i/F_{i-1}$ and $\Gamma_i=\Gamma(E_i)$. The
discriminant of $F$ is given in terms of the subobjects and quotients by:
\be
\label{eq:discriminant}
\Delta(\,\Gamma(F)\,)=\sum_{i=1}^\ell\frac{r(E_i)}{r(F)}\Delta(E_i)-\frac{1}{2r(F)}\sum_{i=2}^\ell
\frac{r(F_i)\,r(F_{i-1})}{r(E_i)} \left(\mu(F_{i})-\mu(F_{i-1}) \right)^2.
\ee

We are interested in the moduli space (or moduli stack) of semi-stable
sheaves with respect to Gieseker stability, but also the coarser
$\mu$-stability appears in order to apply the blow-up formula. To
define these two stability conditions, let $C(S)\subset H^2(S,\mathbb{R})$ be
the ample cone of $S$, and the (reduced) Hilbert polynomial $p_{J}(F,n)=
\chi(F\otimes J^n)/r(F)$. For a surface $S$, we have \cite{Friedman:1998}:
\be
p_{J}(F,n)=J^2n^2/2+\left( \frac{c_1(F)\cdot J}{r(F)}-\frac{K_S\cdot
    J}{2}\right)n +\frac{1}{r(F)}\left(\frac{c_1(F)^2-K_S\cdot c_1(F)}{2}-c_2(F) \right)+\chi(\CO_S).
\ee
Note that this function can be obtained from the physical central
charge as in \cite{Diaconescu:2007bf, Manschot:2009ia}. In the large volume limit, the
stability condition asymptotes
to the lexicographic ordering of polynomials based on
their coefficients. This ordering is denoted by $\prec$.  
Then,
\begin{definition}
A torsion free sheaf $F$ is Gieseker stable (respectively semi-stable) if for every subsheaf $F'\subsetneq
F$, $p_J(F',n)\prec p_J(F,n)$ ( respectively $p_J(F',n)\preceq
p_J(F,n)$ ).
\end{definition}
and
\begin{definition}
Given a choice $J\in C(S)$, a torsion free sheaf $F$ is called
$\mu$-stable if for every subsheaf $F'\subset F$,
$\mu(F')\cdot J <\mu(F)\cdot J$, and $\mu$-semi-stable if for every subsheaf $F'$,
$\mu(F')\cdot J \leq\mu(F)\cdot J$. 
\end{definition}
Thus $\mu$-stability is a coarser stability condition then Gieseker
stability, although the walls of marginal stability for both stability conditions are the
same. A wall of marginal stability $W(F',F)\subset H^2(S,\mathbb{R})$ is the
codimension 1 subspace of $C(S)$, such that $(\mu(F')-\mu(F))\cdot
J=0$, but $(\mu(F')-\mu(F))\cdot J\neq 0$ away from $W(F',F)$. 
The invariants based on Gieseker stability exhibit better
integrality and polynomial properties then the ones based on $\mu$-stability. On the other hand, operations like restriction
to a curve and blowing-up a point of $S$ are most natural for $\mu$-semi-stable sheaves. 

The moduli space $\CM_J(\Gamma)$ of Gieseker stable sheaves on $S$ (with respect
to the ample class $J$) whose rank and Chern classes are determined by
$\Gamma$ has expected dimension: 
\be
\label{eq:dim}
d_{\mathrm{exp}}(\Gamma)=\dim_{\mathbb{C}}(\Ext^1(F,F))-\dim_{\mathbb{C}}(\Ext^2(F,F))=2r^2\Delta-r^2\chi(\CO_S)+1.
\ee
When $\Ext^2(F,F)=0$ the moduli space is smooth and of the
expected dimension. Vanishing of $\Ext^2(F,F)$ for
semi-stable sheaves on surfaces can be proven if the polarization
satisfies $J\cdot K_S<0$. More generally, we have
\begin{proposition}
Let $J\in C(S)$ such that $J\cdot K_S<0$ and let $F$ and $G$ be Gieseker
semi-stable sheaves with respect to polarization $J$ such that $p_J(F,n) \preceq p_J(G,n)$. Then:
$$\mathrm{Ext}^2(F,G)=0.$$ 
\end{proposition}
\begin{proof}
Due to Serre duality $\mathrm{Ext}^2(F,G)=\mathrm{Hom}(G,F\otimes K_{S})^\vee$.
Assume contrary to the proposition that $\mathrm{Ext}^2(F,G)\neq 0$,
such that a non-vanishing morphism $\psi:G\to F \otimes K_{S}$
exists. Then $F \otimes K_{S}$ is a quotient of $G$, and
semi-stability of $G$ implies $p_J(F\otimes K_{S},n)\succeq
p_J(G,n)$. Now we find a contradiction, since the assumption $J\cdot
K_S<0$ implies $p_J(F\otimes
K_{S},n) \prec p_J(F,n) \preceq p_J(G,n)$. Therefore a non-vanishing
$\psi$ cannot exist and the proposition follows. 
\end{proof}
\noi Dimension estimates for (coarse) moduli spaces of
semi-stable sheaves are more subtle due to endomorphisms. 
We will find that BPS invariants computed in Sections \ref{sec:ruledsurface} and
\ref{sec:projplane}  are in agreement with the expected dimension (if non-vanishing).

Twisting a sheaf $E$ by a line bundle $\CL$ is an isomorphism of moduli
spaces. The Chern classes of the twisted sheaf $E'=E\otimes \CL$ are:
\begin{eqnarray}
\label{eq:ltwist}
&& r(E')=r(E), \quad c_1(E')=c_1(E)+r(E)c_1(\CL),\non\\
&& c_2(E')=c_2(E)+(r(E)-1)c_1(\CL)c_1(E)+c_1(\CL)^2\frac{r(E)(r(E)-1)}{2}.\non
\end{eqnarray}
The discriminant remains invariant: $\Delta(E')=\Delta(E)$. This shows that it suffices to compute
the generating functions for $c_1(E)\in H^2(S,\mathbb{Z}/r\mathbb{Z})$.

Determination of generating functions of BPS invariants for $r\geq 2$
is an open problem in general. To make progress, we specialize in the following to the
set of smooth ruled surfaces. A ruled surface is a surface $\Sigma_{g,\ell}$ together with a surjective
morphism $\pi: \Sigma_{g,\ell} \to C_g$ to a curve $C_g$ with genus $g$,
such that the fibre over each point of $C_g$ is a smooth irreducible
rational curve and such that $\pi$ has a section. Let $f$ be the
fibre of $\pi$, then $H_2(\Sigma_{g,\ell},\mathbb{Z})=\mathbb{Z}C_g\oplus\mathbb{Z}f$, with 
intersection numbers $C_g^2=-\ell$, $f^2=0$ and $C_g\cdot f=1$. The canonical class is
$K_{\Sigma_{g,\ell}}=-2C_g+(2g-2-\ell)f$. The holomorphic Euler characteristic
$\chi(\CO_{\Sigma_{g,\ell}})$ is $1-g$. An ample
divisor $J\in C(\Sigma_{g,\ell})$ is parametrized by $J_{m,n}=m(C_g+\ell
f)+nf$ with $m,n> 0$. The condition $J\cdot K_S<0$ translates to
$m(2g-2-\ell)<2n$.

Most of this article will further specialize to the Hirzebruch surfaces
$\Sigma_{0,\ell}=\Sigma_\ell$. For these surfaces $J\cdot K_S<0$ is
satisfied for all $J\in C(\Sigma_{\ell})$. The surface $\Sigma_1$
playes a special role since besides being a ruled surface,
$\Sigma_{1}$ is also the blow-up $\phi: \mathbb{\tilde P}^2\to
\mathbb{P}^2$ of the projective plane $\mathbb{P}^2$. The exceptional divisor of $\phi$ is $C_0=C$, and
the pullback of the hyperplane class $H$ of $\bP^2$ is given by $\phi^*H=C+f$. Due to the
simplicity of $\bP^2$, it is of intrinsic interest to determine the
generating functions of its BPS invariants.

\section{BPS invariants and generating functions}
\label{sec:genfunctions}
This section defines the generating functions of the BPS invariants
and discusses some of its properties. Physically, the BPS invariant
arises by considering topologically twisted $\CN=4$ Yang-Mills on the
surface $S$ \cite{Vafa:1994tf}. The path integral of this theory localizes on the
BPS solutions, including the instantons, due to the topologically
twisted supersymmetry \cite{Vafa:1994tf}. The BPS invariant is given by a weighted sum over the BPS
Hilbert space $\CH(\Gamma,J)$, and based on the path integral one can show that the (numerical) BPS invariant corresponds to the Euler
number of the BPS moduli space.  

Alternatively one can consider the $\CN=2$ supersymmetric 
theory in $\mathbb{R}^{3,1}$ obtained from the compactification of
IIA theory on a non-compact Calabi-Yau $\mathcal{O}(-K_S)\to
S$. The $\CN=2$ theory with gauge group $SU(2)$ and without hypermultiplets can be engineered by any of the
Hirzebruch surfaces $\Sigma_\ell$ \cite{Katz:1996fh}. Sheaves
supported on $\Sigma_\ell$ correspond to magnetic monopoles and
dyons in $\CN=2$ gauge theory. In this theory, the BPS invariant can be refined with
an additional parameter $w$ \cite{Gaiotto:2010be}:
\be
\Omega(\Gamma,w;J)=\frac{\mathrm{Tr}_{\CH(\Gamma,J)}\, 2\hat J_3(-1)^{2\hat
    J_3}(-w)^{2\hat I_3 +2\hat J_3}}{(w-w^{-1})^2},
\ee
with $\hat J_3$ a generator of the $SU(2)\cong \mathrm{Spin}(3)$ group
arising from rotations in $\mathbb{R}^{3,1}$, and $\hat I_3$ is a generator of the $SU(2)_R$
$R$-symmetry group. BPS representations have the form $\left[ (\half,0)\oplus
(0,\half)\right]\otimes \omega$ with $\omega=(j,j')$ a vacuum
representation of $\mathrm{Spin}(3)\oplus
SU(2)_R$ with spins $j$ and $j'$. One factor of $w-w^{-1}$ in the
denominator will vanish due to the factor $(\half,0)\oplus (0,\half)$ (the
half-hypermultiplet) present for every BPS
state \cite{Gaiotto:2010be}.  Since $\Omega(\Gamma,w;J)$ is thus essentially an $SU(2)$
character, this shows that $\Omega(\Gamma,w;J)$ is a polynomial divided by
$w-w^{-1}$; the polynomial has integer coefficients and is invariant under $w\leftrightarrow
w^{-1}$. The positivity conjectures of Ref. \cite{Gaiotto:2010be} assert furthermore
that the coefficients are positive. We choose to divide by the factor
$w-w^{-1}$ in order to have nice modular properties of the
generating functions. See for example Eq. (\ref{eq:rank1}). 

The $\CN=2$ picture shows that the refined BPS invariant provides more
information than the Euler number of the moduli space. The
$w$-expansion is expected to give the $\chi_y-$genus of the BPS moduli
space \cite{Chuang:2013}. To make this more precise, we let $\CM_J(\Gamma)$ be the suitably compactified moduli space of
semi-stable sheaves on $S$ with topological classes $\Gamma$ and for
polarization $J\in C(S)$, i.e. the Gieseker-Maruyama
compactification. If we assume that $J\cdot K_S<0$ and that
semi-stable is equivalent to stable, the moduli space is smooth and
the BPS invariant corresponds mathematically to \cite{Chuang:2013}:
\be
\label{eq:BPSinvariant}
\Omega(\Gamma,w;J):=\frac{w^{-\dim_\mathbb{C}\CM_J(\Gamma)}}{w-w^{-1}}\, 
\chi_{w^2}(\CM_J(\Gamma)), \qquad w^2\neq 1,
\ee
with on the right hand side the $\chi_y$-genus, which is defined in
terms of the virtual Hodge numbers
$h^{p,q}(X)=\dim H^{p,q}(X,\mathbb{Z})$ of the quasi-projective
variety $X$ by $\chi_{y}(X)=\sum_{p,q=
  0}^{\dim_\mathbb{C}(X)}(-1)^{p-q}\,y^p\,h^{p,q}(X)$. Eq. (\ref{eq:dim}) provides us with the
degree of $\chi_{w^2}(\CM_J(\Gamma))$, and since $\CM_J(\Gamma)$ is compact, orientable
and without boundary $h^{p,q}(X)=h^{\dim_\mathbb{C}(X)-p,\dim_\mathbb{C}(X)-q}(X)$. 
 For rational
surfaces, which include the ruled surfaces with $g=0$, the non-vanishing cohomology
of smooth moduli spaces of semi-stable sheaves has Hodge type $(p,p)$ \cite{Beauville:1992,
  Gottsche:1998}. Therefore, $\chi_{w^2}(X)=P(X,w)=\sum_{i=0}^{2\dim_\mathbb{C}(X)}b_i(X)\,w^i$
with $P(X,w)$ the Poincar\'e polynomial and
$b_i(X)=\sum_{p+q=i}h^{p,q}(X)$ the Betti numbers of $X$. 

If semi-stable is not equivalent to stable, $\CM_J(\Gamma)$ contains
singularities due to non-trivial automorphisms of the sheaves. The
formal mathematical framework for the integer BPS invariants or
motivic Donaldson-Thomas invariants  is developed by Kontsevich and Soibelman \cite{Kontsevich:2008}. For our
purposes it is useful to introduce also two other invariants, $\bar\Omega(\Gamma,w;J)$ and
$\CI(\Gamma,w;J)$. 
These invariants are defined using the notion of
moduli stack $\mathfrak{M}_J(\Gamma)$ which properly deals with the
mentioned singularities in the moduli space of semi-stable sheaves by
keeping track of the automorphism groups of the semi-stable sheaves.

The invariant $\CI(\Gamma,w;J)$ is an example of a motivic
invariant. In general an invariant of a quasi-projective variety $X$ is called 'motivic'
if $\Upsilon(X)$ satisfies:
\begin{itemize}
\item[-] If $Y\subseteq X$ is a closed subset then $\Upsilon(X)=\Upsilon (X\,\backslash\, Y) + \Upsilon (Y)$,
\item[-] If $X$ and $Y$ are quasi-projective varieties  $\Upsilon(X\times Y)=\Upsilon(X)\, \Upsilon(Y)$.
\end{itemize}
Ref. \cite{Joyce:2005} defines a motivic invariant, the virtual
Poincar\'e function $\Upsilon'$, for Artin stacks,
which are stacks whose stabilizer groups are algebraic groups. The virtual Poincar\'e function
$\CI(\Gamma,w;J)$ is a rational function in $w$ and a natural 
generalization of the Poincar\'e polynomial of smooth projective
varieties to stacks. The definition of these invariants for stacks is such
that for a quotient stack $[X/G]$ with $G$ an algebraic group, one
has $\Upsilon'([X/G])=\Upsilon(X)/\Upsilon(G)$. 

Using the virtual Poincar\'e function $\Upsilon'$, Definition 6.20 of Ref. \cite{Joyce:2004}
defines the virtual Poincar\'e function $\CI(\Gamma,w;J)$ (in Ref. \cite{Joyce:2004} denoted by
$I^\alpha_{\mathrm{ss}}(\tau)^\Lambda$) for the moduli stacks of
semi-stable sheaves on surfaces with Ext$^2(X,Y)=0$ for
$p_J(X,n)\prec p_J(Y,n)$.  Definition 6.22 of Ref. \cite{Joyce:2004} also
defines a second invariant $\bar \Omega(\Gamma,w;J)$
(denoted by $\bar J^\alpha(\gamma)^\Lambda$ in
Ref. \cite{Joyce:2004}). These appear in fact  rather natural from the physical perspective \cite{Manschot:2010qz, Kim:2011sc}.  See also
\cite{Nakajima:2010} for related discussions of invariants.

The invariants $\bar \Omega(\Gamma,w;J)$ are the rational multi-cover
invariants of $\Omega(\Gamma,w;J)$:
\begin{eqnarray}
\label{eq:refw}
\bar \Omega(\Gamma,w;J)&:=&\sum_{m|\Gamma}
\frac{\Omega(\Gamma/m,-(-w)^m;J)}{m}.
\end{eqnarray}
They can be expressed in terms of $\mathcal{I}(\Gamma,w;J)$ and vice versa
 (Theorem 6.8 in \cite{Joyce:2004}):
\be
\label{eq:inversestackinv}
\bar \Omega(\Gamma_i,w;J):=\sum_{\Gamma_1+\dots +\Gamma_\ell=\Gamma\atop
p_J(\Gamma_i,n)=p_J(\Gamma,n)\,\mathrm{for}\,\, i=1,\dots,\ell}
\frac{(-1)^{\ell+1}}{\ell} \,\prod_{i=1}^\ell \mathcal{I}(\Gamma_i,w;J).
 \ee 
with inverse relation:
\be
\label{eq:stackinvariant}
\mathcal{I}(\Gamma,w;J)=\sum_{\Gamma_1+\dots +\Gamma_\ell=\Gamma\atop
p_J(\Gamma_i,n)=p_J(\Gamma,n)\,\mathrm{for}\,\, i=1,\dots,\ell}
\frac{1}{\ell !} \,\prod_{i=1}^\ell \bar \Omega(\Gamma_i,w;J),
\ee 
Note that $\mathcal{I}(\Gamma,w^{-1};J)\neq -\mathcal{I}(\Gamma,w;J)$
and that $\mathcal{I}(\Gamma,w;J)$ in general has higher order poles in
$w$ compared to $\Omega(\Gamma,w;J)$. 

It is an interesting question what geometric information the integer
invariants $\Omega(\Gamma,w;J)$ carry if $m|\Gamma$ with $m>1$.
For $r=m=2$, Remark 4.6 of Ref. \cite{Yoshioka:1995} argues that $\Omega(\Gamma,w;J)$ computes the Betti 
numbers of rational intersection cohomology of
the singular moduli space $\CM_J(\Gamma)$. The generating
function in Remark 4.6 of Ref. \cite{Yoshioka:1995} is very closely
related to the one obtained for moduli spaces of semi-stable
vector bundles over Riemann surfaces in (the Corrigendum to) Ref. \cite{Kirwan:1986}. Intersection cohomology is
a cohomology theory for manifolds with singularities which
satisfies Poincar\'e duality if the manifolds are complex and compact. It is therefore natural to expect that
the BPS invariant (\ref{eq:BPSinvariant}) for $r\geq 3$ also provides Betti numbers of
intersection cohomology groups. This issue is left for further
research. 

The seminal papers \cite{Gottsche:1990, Yoshioka:1994, Yoshioka:1995}
compute moduli space and stack invariants by explicitly counting sheaves on
the surface $S$ defined over a finite field $\mathbb{F}_s$ with $s$
elements. The Poincar\'e function $\CI(\Gamma,s^\frac{1}{2};J)$ is upto an overall monomial computed by:
\be
\label{eq:countoverF}
\sum_{E\in M_J(\Gamma,\bF_s)} \frac{1}{\#\mathrm{Aut}(E)},
\ee
where $M_J(\Gamma,\bF_s)$ is the set of semi-stable sheaves with
characteristic classes $\Gamma$. The Weil  
conjectures imply that the expansion coefficients in $s$ are the Betti numbers
of the moduli spaces. The parameter $s$ is related to the $w$ in this
article by $s=w^{2}$. Eq. (\ref{eq:countoverF}) shows that poles of
$\CI(\Gamma,w;J)$ in $w$ appear when the sheaves have non-trivial automorphism groups. If semi-stable is  
equivalent to stable $\CI(\Gamma,w;J)=\Omega(\Gamma,w;J)$; the factor
$(w-w^{-1})^{-1}$ in Eq. (\ref{eq:BPSinvariant}) is due to the automorphisms which are multiplication by
$\mathbb{C}^*$. The automorphism group of semi-stable and unstable
bundles or sheaves is in general $GL(n)$, whose number of elements over $\mathbb{F}_s$
is $(1-s)(1-s^2)\dots (1-s^n)$ and thus lead to higher order poles.


We continue now by defining the generating function
$h_{r,c_1}(z,\tau;S,J)$ of $\bar \Omega(\Gamma,w;J)$:
\be
h_{r,c_1}(z,\tau;S,J)=\sum_{c_2} \bar \Omega(\Gamma,w;J)\,q^{r\Delta(\Gamma)-\frac{r\chi(S)}{24}}.
\ee
where $q:=e^{2\pi i \tau}$, with $\tau \in \mathcal{H}$ and $w:=
e^{2\pi i z}$ with $z\in \mathbb{C}$. Since twisting by a line bundle
(\ref{eq:ltwist}) is an isomorphism of moduli
spaces, it suffices to compute $h_{r,c_1}(z,\tau;S,J)$ for $c_1\in
H_2(S,\mathbb{Z}/r\mathbb{Z})$. The expansion parameter $t$ for $c_2$
in Refs. \cite{Gottsche:1990, Yoshioka:1994, Yoshioka:1995} is related
to $q$ by $q=s^rt$.
 
The generating function $h_{1,c_1}(z,\tau;S)$ depends only on 
$b_2(S)$ for $S$ a smooth projective surface with $b_{1}(S)=b_{3}(S)=0$ \cite{Gottsche:1990}: 
\be
\label{eq:rank1}
h_{1,c_1}(z,\tau;S)=\frac{i}{\theta_1(2z,\tau)\,\eta(\tau)^{b_2(S)-1}},
\ee
where the Dedekind eta function $\eta(\tau)$ and Jacobi theta function
$\theta_1(z,\tau)$ are defined by:  
\begin{eqnarray}
\label{eq:etatheta}
\eta(\tau)\quad \,\,&:=&q^{\frac{1}{24}}\prod_{n=1}^\infty (1-q^n),\non\\
\theta_1(z,\tau)&:=&i q^{\frac{1}{8}}(w^\frac{1}{2}-w^{-\frac{1}{2}})\prod_{n\geq 1}(1-q^n)(1-wq^n)(1-w^{-1}q^n).\non
\end{eqnarray}
The dependence on $J$ is omitted in Eq. (\ref{eq:rank1}), since all
rank 1 torsion free sheaves are stable throughout $C(S)$. 
Similarly, $J$ is omitted in the following from
$h_{r,c_1}(z,\tau;\bP^2,J)$, since $b_2(\bP^2)=1$ and therefore the BPS invariants do not vary as
function of $J$. For clarity of exposition, $\Sigma_\ell$ is omitted from the arguments 
of $h_{r,c_1}(z,\tau;\Sigma_\ell,J)$.

We will be mainly concerned with the invariants $\bar
\Omega(\Gamma,w;J)$ since the generating functions are defined in
terms of these invariants. However, some formulas are most naturally phrased in terms
of $\mathcal{I}(\Gamma,w;J)$. For example, the product formula of Conjecture \ref{conj:restrictfibre}
is a generating function for $\mathcal{I}(\Gamma,w;f)$ and the blow-up
formula in Section \ref{sec:projplane} is phrased in terms of
$\mathcal{I}^\mu(\Gamma,w;J)$, which are invariants with respect to $\mu$-stability instead of Gieseker stability.


\section{Restriction to the fibre of Hirzebruch surfaces}
\label{sec:restrictfibre}
This subsection deals with the set $M_f(\Gamma)$ of sheaves whose restriction to the (generic)
fibre $f$ of $\pi:\Sigma_\ell \to C$ is semi-stable. Inspired by the
existing results for $r=1$ and 2 \cite{Gottsche:1990, Yoshioka:1995}
and moduli stack invariants for vector bundles over Riemann surfaces
\cite{Harder:1975, Atiyah:1982fa}, a generating
function for $r\geq 1$ is proposed enumerating virtual Poincar\'e functions $\CI(\Gamma,w;f)$ of moduli 
stacks $\mathfrak{M}_f(\Gamma)$ of sheaves whose restriction to the
fibre is semi-stable. We do not present a derivation of this generating
function based on $\mathfrak{M}_f(\Gamma)$ for $r\geq 3$, nor an analysis of the
properties of $\mathfrak{M}_f(\Gamma)$. Section \ref{sec:ruledsurface} computes the BPS invariants
starting from these generating functions, and shows that they pass
various non-trivial consistency checks implied by the blow-up 
and wall-crossing formulas. 

We define the generating function $H_{r,c_1}(z,\tau;f)$ of
$\CI(\Gamma,w;f)$ by:
\be
H_{r,c_1}(z,\tau;f):=\sum_{c_2} \CI(\Gamma,w;f)\,
q^{r\Delta(\Gamma)-\frac{\chi(S)}{24}}.
\ee
The following conjecture gives $H_{r,c_1}(z,\tau;f)$ for any $r\geq 1$ and
$c_1\in H_2(\Sigma_\ell,\mathbb{Z})$: 
\begin{conjecture} 
\label{conj:restrictfibre}
The function $H_{r,c_1}(z,\tau;f)$ is given by:
\be
\label{eq:restrictfibre}
H_{r,c_1}(z,\tau;f)=\left\{ \begin{array}{cl}
     \frac{i\,(-1)^{r-1}\,\eta(\tau)^{2r-3}}{\theta_1(2z,\tau)^2\,\theta_1(4z,\tau)^2\dots\theta_1((2r-2)z,\tau)^2\,\theta_1(2rz,\tau)},
    & \mathrm{if}\,\,c_1\cdot f=0\mod r,\quad r\geq 1, \\ 0, &
    \mathrm{if}\,\,c_1\cdot f\neq 0\mod r, \quad r>1. \end{array}\right.
\ee
\end{conjecture}
The above expressions for $H_{r,c_1}(z,\tau;f) $ are not conjectural
for all $(r,c_1)$. Vanishing of $H_{r,c_1}(z,\tau;f)$ for $c_1\cdot f\neq 0\mod r$ is well
known. See for example Section 5.3 of \cite{Huybrechts:1996}. The vanishing is
a consequence of the fact that all bundles $F$ on $\bP^1$ are
isomorphic to a sum of line bundles $F\cong \CO(d_1)\oplus \CO(d_2)\dots
\CO(d_r)$. Therefore, a bundle $F$ on $\bP^1$ can only be semi-stable\footnote{Recall that a
  vector bundle $F$ of rank $r$ and degree $d$ on a curve $C$ is stable (respectively semi-stable) if for every
subbundle $F'\subsetneq F$ (with rank $r'$ and degree $d'$)
$d'/r'<d/r$ (respectively $d'/r'\leq d/r$).} if its degree $d$ is equal to $ 0 \mod r$ such that the degrees of the
line bundles are $d_i=d/r$.  The degree $d(E_{|f})$ of the restriction of a sheaf
$E$ on $\Sigma_\ell$ to $f$ is equal to $c_1(E)\cdot
f$. Therefore, the only cases for which $H_{r,c_1}(z,\tau;f)$ does not
vanish is for $c_1\cdot f=0\mod r$. 

For $r=1$, Eq. (\ref{eq:restrictfibre}) reduces to Eq. (\ref{eq:rank1}).
Ref. \cite{Yoshioka:1995} proved the conjecture for $(r,c_1)=(2,f)$,
which is now briefly recalled. Ref. \cite{Yoshioka:1995} considers the
ruled surface $\tilde \bP^2$ over a finite field $\mathbb{F}_s$, and utilizes the fact that any vector bundle in $F$ 
can be obtained from $\pi^*\pi_* F$, which is a vector bundle on
$\tilde \bP^2$ supported on $C$, by successive elementary
transformations. 

An elementary transformation is defined by \cite{Huybrechts:1996}:
\begin{definition}
Let $D$ be an effective divisor on the surface $S$. If $F$ and $G$ are
vector bundles on $S$ and $D$ respectively, then a vector bundle $F'$
on $S$ is obtained by an elementary transformation of $F$ along $G$ if
there exists an exact sequence:
\be
0\to F' \to F\to i_* G\to 0,
\ee
where $i$ denotes the embedding $D\subset S$.
\end{definition}  

This shows that the contribution to $h_{2,c_1}(z,\tau;J)$ from $M_f(\Gamma)$ is
the product of the total set of vector bundles on $C$, multiplied
by the number of elementary transformations. The total set of vector
bundles with $r=2$ on $C$ is enumerated by \cite{Harder:1975}:
\be
\label{eq:zeta2}
\frac{s^{-3}}{1-s}\,\zeta_C(2)
\ee
where $\zeta_{C}(n)$ is the zeta function of the Riemann surface
$C_0$. One has for general genus $g$:
\be
\zeta_{C_g}(n)=\frac{\prod_{j=1}^{2g}(1-\omega_j s^{-n})}{(1-s^{-n})(1-s^{1-n})}.
\ee
Multiplication of (\ref{eq:zeta2}) by the factor due to elementary transformations gives \cite{Yoshioka:1995}: 
\be
\label{eq:setrestrict}
\sum_{c_2} \sum_{E\in M_f(2,mf,c_2)}\frac{t^{c_2}}{\#
  \Aut(E)}=\frac{s^{-3}}{1-s}\zeta_C(2)\,\prod_{a\geq 1}Z_s(S,s^{2a-2}t^a) Z_s(S,s^{2a}t^a),
\ee
with $Z_s(S,t)$ the zeta function of the surface $S$:
\be
Z_s(S,t)=\frac{1}{(1-t)(1-st)^{b_2(S)}(1-s^2t)}.
\ee
The parameter substitutions
$q=s^rt$ and $w^2=s$ give then Eq. (\ref{eq:restrictfibre}) (upto
an overal monomial in $w$ and $q$).  

This derivation for $r=2$ indicates that $H_{r,c_1}(z,\tau;f)$ is closely related to
that of the virtual Poincar\'e function of the stack of vector bundles on a Riemann
surface $C_g$ with genus $g$ \cite{Harder:1975, Atiyah:1982fa}: 
\be
\label{eq:totalset}
H_r(z;C_g):=-w^{r^2(1-g)}\frac{(1+w^{2r-1})^{2g}}{1-w^{2r}}\prod_{j=1}^{r-1}\frac{(1+w^{2j-1})^{2g}}{(1-w^{2j})^{2}}.
\ee
The first term in the $q$-expansion of  Eq. (\ref{eq:restrictfibre}) starts with
Eq. (\ref{eq:totalset}) for $g=0$. One could thus understand $H_{r,c_1}(z,\tau;f)$ as
an extension of $H_r(z;C_0)$ to a modular infinite product. 
It is conceivable that Conjecture \ref{conj:restrictfibre} for
$r>2$ can be proven in a similar manner as for $r=2$. The following sections show
that at least for $r=3,4$, it is consistent with various other
results. Moreover, it continues to hold for the other Hirzebruch
surfaces with $\ell\geq 0$.

As an aside we mention the generalization of the conjecture to ruled surfaces $\Sigma_{g,\ell}$ over a Riemann surface $C_g$ with
$g> 0$. These surfaces are not rational and the moduli
spaces of semi-stable sheaves for these surfaces also have cohomology $H^{p,q}(\CM_J(\Gamma),\mathbb{Z})$ for $p\neq
q$. In order to capture this more refined information we recall the
refinement of Eq. (\ref{eq:totalset}) to the virtual Hodge function \cite{Earl:2000}:
\be
\label{eq:totalset2}
H_r(u,v;C_g):=-\frac{(xy)^{r^2(1-g)/2}}{1-x^{r}y^r}\frac{\prod_{j=1}^{r} (1+x^jy^{j-1})^{g}(1+x^{j-1}y^j)^{g}}{\prod_{k=1}^{r-1} (1-x^{k}y^k)^{2}}.
\ee
with $x:=e^{2\pi i u}$ and $y:=e^{2\pi i v}$. The structure of this
function directly suggests the following generalization of Conjecture
\ref{conj:restrictfibre} for the generating function
$H_{r,c_1}(u,v,\tau;f,\Sigma_{g,\ell})$ of virtual Hodge functions $\CI(\Gamma_i,x,y;f)$ of the moduli
stack $\mathfrak{M}_f(\Gamma;\Sigma_{g,\ell})$
:
\begin{conjecture} 
\label{conj:restrictfibre2}
The function $H_{r,c_1}(u,v,\tau;f,\Sigma_{g,\ell})$ is given by:
\be
\label{eq:restrictfibre2}
\left\{ \begin{array}{cl}
     \frac{i\,(-1)^{r-1}\,\eta(\tau)^{2r(1-g)-3}}{\theta_1(r(u+v),\tau)}
     \frac{\prod_{j=1}^{r} \theta_1(ju +(j-1)v+\frac{1}{2},\tau)^{g}\,
       \theta_1((j-1)u +jv+\frac{1}{2},\tau)^{g}}{\prod_{k=1}^{r-1} \theta_1(k(u+v),\tau)^2},
    & \mathrm{if}\,\,c_1\cdot f=0\mod r,\quad r\geq 1, \\ 0, &
    \mathrm{if}\,\,c_1\cdot f\neq 0\mod r, \quad r>1. \end{array}\right.
\ee
\end{conjecture}

\section{BPS invariants of Hirzebruch surfaces}
\label{sec:ruledsurface}

\subsection{BPS invariants for a suitable polarization}
\label{subsec:suitable}

This subsection computes for $c_1\cdot f= 0 \mod r$ the
BPS invariants of $\Sigma_\ell$ for a polarization $J\in
C(\Sigma_\ell)$ sufficiently close to $J_{0,1}=f$. The BPS invariants
are for this choice of $J$ independent of $\ell$.  
 ``Sufficiently close''  depends on the topological classes of
the sheaf. Generalizing Def. 5.3.1 of \cite{Huybrechts:1996} to
general $r\geq 1$, we define a $\Gamma$-suitable polarization by:
\begin{definition}
\label{def:suitable}
A polarization $J$ is called $\Gamma$-suitable if and only if:
\begin{itemize}
\item[-] $J$ does not lie on a wall for $\Gamma=(r,\mathrm{ch}_1,\mathrm{ch}_2)$ and,
\item[-] for any $J$-semi-stable subsheaf $F'\subset F$ with $\Gamma(F)=\Gamma$, $(\mu(F')-\mu(F))\cdot f=0$ or
$(\mu(F')-\mu(F))\cdot f$ and $(\mu(F')-\mu(F))\cdot J$ have the
same sign.
\end{itemize} 
\end{definition}

We will keep the dependence on the Chern classes implicit in the
following and denote a suitable polarization by $J_{\varepsilon,1}$
with $\varepsilon$ positive but sufficienty small. From the definition follows that if $J_{\varepsilon,1}$ is a $\Gamma(F)$-suitable polarization,
and $F_{|f}$ is unstable, then $F$ is $\mu$-unstable. 
Thus we need to subtract from $M_f(\Gamma)$, i.e. the set of sheaves
with topological classes $\Gamma$ whose restriction to the fibre $f$
is semi-stable, the subset of $M_f(\Gamma)$ which is Gieseker unstable for
$J_{\varepsilon,1}$. We continue by explaining this for
$r=2$. Then the general formula is proposed for the invariant
enumerating extended HN-filtrations, which is consequently applied to $r=3$. 

A crucial tool to obtain the invariants enumerating semi-stable
sheaves are Harder-Narasimhan filtrations \cite{Harder:1975}, which can
be defined for either Gieseker or $\mu$-stability. To define these filtrations, let
$\varphi$ denote either Gieseker, $\varphi(F)=p_J(F,n)$, or $\mu$-stability, $\varphi(F)=\mu(F)\cdot J$. Then:
\begin{definition}\label{def:HNfiltr} A Harder-Narasimhan filtration (HN-filtration) with respect to the
  stability condition $\varphi$ is a filtration $0\subset F_1 \subset F_2\subset \dots \subset F_\ell=F$ of the sheaf
  $F$ such that  the quotients $E_i=F_i/F_{i-1}$ are semi-stable with
  respect to $\varphi$ and satisfy $\varphi(E_i)>\varphi(E_{i+1})$ for all $i$.
\end{definition}
Since $\mu$-stability is coarser then Gieseker stability, the length $\ell_{\mathrm{G}}(F)$
of the HN-filtration with respect to Gieseker stability is in general
larger than the length $\ell_\mu(F)$ of its
HN-filtration with respect to $\mu$-stability.

Using the additive and multiplicative properties of motivic invariants
discussed below (\ref{eq:BPSinvariant}), one can determine the BPS invariants for a suitable polarization.
The Poincar\'e function of the stack of HN-filtrations with respect to
Gieseker stability and prescribed $\Gamma_i=\Gamma(E_i)$ is \cite{Yoshioka:1996}:  
\be
\label{eq:setHN}
w^{-\sum_{i<j} r_ir_j (\mu_j-\mu_i)\cdot K_S} \prod_{i=1}^\ell \CI(\Gamma_i,w;J),
\ee
where $r_ir_j (\mu_j-\mu_i)\cdot K_S$ is the Euler form for
semi-stable sheaves on the projective surface $S$. One could define a
similar function for the stack of filtrations with respect to
$\mu$-stability. For the generalization to Hodge numbers, one replaces
$w^2$ by $xy$ in $w^{-\sum_{i<j} r_ir_j (\mu_j-\mu_i)\cdot K_S}$ and
$\CI(\Gamma_i,w;J)$ by $\CI(\Gamma_i,x,y;J)$.
 
For $(r,c_1)=(2,f)$, the only HN-filtrations with respect to $J_{\varepsilon,1}$ have length $\ell_\mathrm{G}=2$. Denoting
$c_1(E_2)=bC-af$, and thus $c_1(E_1)=-bC+(a+1)f$, one easily verifies 
that the HN-filtrations correspond to $a\geq 0$ and $b=0$. Since $b=0$
the dependence of $K_S$ in Eq. (\ref{eq:setHN}) does not lead to a
dependence on $\ell$. Using that Eq. (\ref{eq:rank1}) is also the generating function of
$\CI(\Gamma,w;J)$ for $r=1$, Eq. (\ref{eq:setHN}) becomes:
\be
\label{eq:contHN1}
\sum_{a \geq 0} w^{-2(2a+1)}\, h_{1,0}(z,\tau)^2 =- \frac{w^2}{1-w^4}\,h_{1,0}(z,\tau)^2,
\ee
where we assumed $|w|>1$. Subtracting this from Eq. (\ref{eq:restrictfibre}) for $r=2$ gives:
\be
h_{2,f}(z,\tau; J_{\varepsilon,1})=\frac{-1}{\theta_1(2z,\tau)^2\,\eta(\tau)^2}\left(\frac{i\,\eta(\tau)^3}{\theta_1(4z,\tau)}+\frac{w^2}{1-w^4}\right),
\ee
which is easily verified to enumerate invariants
$\bar \Omega(\Gamma;J_{\varepsilon,1})$ satisfying the expected properties
mentioned below Eq. (\ref{eq:BPSinvariant}).   

For $(r,c_1)=(2,0)$, the HN-filtrations with respect to $J_{\varepsilon,1}$ and $\ell_\mathrm{G}=2$ split naturally in two
subsets: the first set has length
$\ell_\mu=2$ with respect to $\mu$-stability, and the second set has $\ell_\mu=1$.
Similarly to (\ref{eq:contHN1}), the first set gives rise to: 
\be
\label{eq:r2unset}
-\frac{1}{1-w^4}\,h_{1,0}(z,\tau)^2 ,
\ee
and the second set to:
\be
\label{eq:r2secset}
\frac{1}{2}h_{1,0}(z,\tau)^2-\frac{1}{2} \sum_{n\geq 0} \Omega((1,0,n),w)^2\,q^{2n},
\ee
where the second term subtracts from the first the Gieseker
semi-stable sheaves which should not be subtracted from
$H_{2}(z,\tau;f)$. Subtraction of Eqs. (\ref{eq:r2unset}) and
(\ref{eq:r2secset}) from $H_{2}(z,\tau;f)$ gives the generating function of
$\CI(\,(2,0,c_2),w;J)$, which corresponds by Eq. (\ref{eq:stackinvariant}) to:
\be
\label{eq:h20}
h_{2,0}(z,\tau;
J_{\varepsilon,1})=\frac{-1}{\theta_1(2z,\tau)^2\,\eta(\tau)^2}
\left(\frac{i\,\eta(\tau)^3}{\theta_1(4z,\tau)}+\frac{1}{1-w^4}-\frac{1}{2}\right),
\ee
Again one can verify that the invariants satisfy the expected
integrality properties. Remark 4.6 of Ref. \cite{Yoshioka:1995} determines the Betti numbers of the
intersection cohomology of the singular moduli spaces and arrives at
the same generating function (\ref{eq:h20}). 

The Betti numbers for the intersection cohomology of the moduli space
of semi-stable vector bundles on Riemann surfaces were earlier computed in
Ref. \cite{Kirwan:1986}. The above procedure gives these Betti numbers
with much less effort. For example, one can easily verify that 
\be
H_2(z,C_g)+\left(\frac{1}{1-w^4}-\frac{1}{2} \right)H_1(z,C_g)^2,
\ee
with $H_r(z,C_g)$ as in Eq. (\ref{eq:totalset}), is equivalent with
Proposition 5.9 in the Corrigendum to \cite{Kirwan:1986}.

%

Since the invariants $\CI(\Gamma,w;J)$ are not so compatible with 
modular generating functions for $r\geq 2$, it is useful to work as much as possible
with the invariants $\bar \Omega(\Gamma,w;J)$. To this end an
extension of the HN-filtration is necessary: 
\begin{definition}
\label{def:extHNfiltr}
An {\rm extended} Harder-Narasimhan filtration (with respect to Gieseker stability) is a filtration $0\subset F_1\subset F_2\subset
\dots \subset F_\ell=F$ whose quotients $E_i=F_i/F_{i-1}$ are
semi-stable and satisfy $p_J(E_i,n)\succeq p_J(E_{i+1},n)$.
\end{definition}

An example of an extended Harder-Narashimhan filtration can be
obtained by considering a Jordan-H\"older filtration of the semi-stable
quotients of a standard HN-filtration. Recall that a Jordan-H\"older
filtrations is a filtration $0\subset F_1\subset F_2\subset
\dots \subset F_\ell=F$ of a semi-stable bundle $F$ such that the
quotients $E_i=F_i/F_{i-1}$ are stable and satisfy
$p_J(E_i,n)=p_J(F,n)$. However, not all extended HN-filtrations are
obtained this way since Definition \ref{def:extHNfiltr} allows for
semi-stable quotients.

From Eq. (\ref{eq:stackinvariant}) follows that the natural invariant
$\bar \Omega(\{ \Gamma_i \},w; J)$ associated to the stack
$\mathfrak{M}_J(\{\Gamma_i\})$ of extended HN-filtrations with
prescribed Chern classes $\Gamma_i=\Gamma(E_i)$ is: 
\be 
\label{eq:setfiltration}
\bar \Omega(\{ \Gamma_i \};w, J):=\frac{1}{|\mathrm{Aut} (\{\Gamma_i \};J) | }w^{-\sum_{i<j} r_ir_j
(\mu_j-\mu_i)\cdot K_S} \prod_{i=1}^\ell \bar \Omega(\Gamma_i,w; J).
\ee
The number $|\mathrm{Aut} (\{\Gamma_i \};J)|$ is equal to
$\prod_am_a!\,$, where $m_a$ is the total number of quotients $E_i$
with equal reduced Hilbert polynomial
$p_J(E_{a},n)$. Thus only for HN-filtrations $|\mathrm{Aut} (\{\Gamma_i
\};J)|=1$.

If the sum over all extended HN-filtrations contains a group $\{ E_i \}$ with equal $p_J(E_{i},n)$ but unequal $\Gamma_i$,
the factor $\frac{1}{|\mathrm{Aut} (\{\Gamma_i \};J)|}$ divides out a
number of permutations. To avoid this overcounting,
one could introduce a further ordering on the vectors $\Gamma_i$,
which should be obeyed by the set of filtrations to be summed over. Then
one would divide by $|\mathrm{Aut}(\{\Gamma_i\})|=\prod_{p} n_p!$,
where $n_p$ is the number of equal vectors $\Gamma_p$ appearing among
the $\Gamma_i$, $i=1,\dots,\ell$. This is the origin of  the ``Boltzmann statistics''
in wall-crossing formulas \cite{Manschot:2010qz} in the work of Joyce \cite{Joyce:2004}. 

The functions $h_{r,c_1}(z,\tau;J_{\varepsilon,1})$ with $c_1\cdot
f=0\mod r$ are given by the recursive formula
\be
\label{eq:recursion}
h_{r,c_1}(z,\tau;J_{\varepsilon,1})=H_{r,c_1}(z,\tau;f)-\sum_{\mathrm{ch}_2}\sum_{
  \Gamma_1+\dots+\Gamma_\ell=(r,c_1,\mathrm{ch}_2) \atop p_J(\Gamma_i,n)\succeq p_J(\Gamma_{i+1},n),\, \ell>1} 
\bar \Omega(\{ \Gamma_i \};w,
J_{\varepsilon,1})\,q^{r\Delta(\Gamma)-\frac{r\chi(S)}{24}},
\ee
with $\Delta(\Gamma)$ given in terms of $\Gamma_i$ by
Eq. (\ref{eq:discriminant}) and $H_{r,c_1}(z,\tau;f)$ defined by Eq. (\ref{eq:restrictfibre}). 

We continue by applying Eq. (\ref{eq:setfiltration}) to compute
$h_{3,c_1}(z,\tau;J_{1,\varepsilon})$, with $c_1=f$ and $0$. One obtains:   
\begin{proposition}
\begin{eqnarray}
\label{eq:h31eps}
h_{3,f}(z,\tau;J_{\varepsilon,1})&=&\frac{i\,\eta(\tau)^3}{\theta_1(2z,\tau)^2\,\theta_1(4z,\tau)^2\,\theta_1(6z,\tau)}
+\frac{w^2+w^4}{1-w^6}\frac{1}{\theta_1(2z,\tau)^3\,\theta_1(4z,\tau)}  \\
&&-\frac{w^4}{(1-w^4)^2}\frac{i}{\theta_1(2z,\tau)^3\eta(\tau)^3}, \non
\end{eqnarray}
\begin{eqnarray}
\label{eq:h30eps}
h_{3,0}(z,\tau;J_{\varepsilon,1})&=&\frac{i\,\eta(\tau)^3}{\theta_1(2z,\tau)^2\,\theta_1(4z,\tau)^2\,\theta_1(6z,\tau)}
+\frac{1+w^6}{1-w^6}\,\frac{1}{\theta_1(2z,\tau)^3\,\theta(4z,\tau)}\\
&& -
\left( \frac{w^4}{(1-w^4)^2}+\frac{1}{3}\right)\,\frac{i}{\theta_1(2z,\tau)^3\,\eta(\tau)^3}. \non 
\end{eqnarray}

\end{proposition}

\begin{proof} We start by proving Eq. (\ref{eq:h31eps}). Denote the
  length of an extended HN-filtration by $\ell$, its length with respect to $\mu$-stability
by $\ell_\mu$ and Gieseker stability $\ell_\mathrm{G}$. We first consider the unstable filtrations
with $\ell=\ell_\mu=2$, and parametrize $c_1(E_2)$ by $bC-af$. These are parametrized by $a\geq 0$
and $b=0$. There are four possibilities to be distinguished: whether $r(E_1)=1$ or 2, and
whether the quotient with rank 2 has $c_1=0$ or $f\mod 2$. Adding up 
these contributions, one obtains:
\be
\label{eq:h31cont1}
-\frac{w^4+w^{8}}{1-w^{12}}\,h_{1,0}(z,\tau)\,h_{2,0}(z,\tau;J_{\varepsilon,1})-\frac{w^2+w^{10}}{1-w^{12}}\,h_{1,0}(z,\tau)\,h_{2,f}(z,\tau;J_{\varepsilon,1}),
\ee
The filtrations with $\ell=3$ consist of 3 subsets: one 
set with $\ell_\mu=3$, one with $\ell_\mu=2$ but $\ell_\mathrm{G}=3$,
and one with $\ell_\mathrm{G}=2$. Parametrizing $c_1(E_i)=b_iC-a_if$, the first set is parametrized by
$a_i-a_{i+1}>0$, $\sum_{i=1}^3 a_i=1$ and $b_i=0$. 
These are counted by:
\be
\label{eq:h31cont2}
\sum_{k_1,k_2>0\atop k_2=k_1-1\mod 3} w^{-4(k_1+k_2)}\,h_{1,0}(z,\tau)^3=\frac{w^4}{(1-w^4) (1-w^{12})}\,h_{1,0}(z,\tau)^3.
\ee
For the second and third sets, one needs to distinguish between
equality of the stability condition of $E_2$ with $E_1$ or
$E_3$. These two sets are enumerated by:
\be
\label{eq:muns}
-\frac{1}{2}\frac{w^4+w^8}{1-w^{12}}\,h_{1,0}(z,\tau)^3.
\ee
Note that the factor
$\frac{1}{|\mathrm{Aut}(\{\Gamma_i\},J_{1,\varepsilon})|}$ naturally
combines the contributions of filtrations with $\ell_\mu<\ell$. Another observation is that the term $-\frac{1}{2}$ in the second
factor of $h_{2,0}(z,\tau;J_{\varepsilon,1})$ (\ref{eq:h20}) cancels against (\ref{eq:muns}) in the total sum. 
After subtraction of the terms (\ref{eq:h31cont1})-(\ref{eq:muns})
from Eq. (\ref{eq:restrictfibre}) for $r=3$, and writing the whole series
in terms of modular functions, one obtains (\ref{eq:h31eps}).  

For $(r,c_1)=(3,0)$, one needs to subtract the following terms:
\begin{itemize}
\item[-] due to unstable filtrations with $\ell=\ell_\mu=2$:
\be
-\frac{2}{1-w^{12}}\, h_{1,0}(z,\tau)\,h_{2,0}(z,\tau;J_{\varepsilon,1})-\frac{2w^6}{1-w^{12}}\, h_{1,0}(z,\tau)\,h_{2,f}(z,\tau;J_{\varepsilon,1}),\non
\ee
\item[-] due to unstable filtrations with $\ell=2$, $\ell_\mu=1$ and
  $\ell_\mathrm{G}=1$ or 2:
\be
\frac{2}{2}\,h_{1,0}(z,\tau)\,h_{2,0}(z,\tau;J_{\varepsilon,1}), \non
\ee
\item[-] due to unstable filtrations with $\ell=\ell_\mu=3$:
\be
\frac{1+w^{12}}{(1-w^8)(1-w^{12})}\,h_{1,0}(z,\tau)^3, \non
\ee
\item[-]  due to unstable filtrations with $\ell=3$, $\ell_\mu=2$ and
  $\ell_\mathrm{G}=2$ or 3:
\be
-\frac{2}{2}\frac{1}{1-w^{12}}\, h_{1,0}(z,\tau)^3, \non
\ee
\item[-] due to unstable filtrations with $\ell=3$, $\ell_\mu=1$ and
  $1\leq \ell_\mathrm{G}\leq 3$:
\be
\frac{1}{6}\, h_{1,0}(z,\tau)^3.\non
\ee
\end{itemize}

Subtracting the terms above from (\ref{eq:restrictfibre}) gives
(\ref{eq:h30eps}). Subtracting further  $\frac{1}{3}h_{1,0}(3z,3\tau)=\frac{i}{3\,\theta_1(6z,3\tau)\,\eta(3\tau)}$ from
(\ref{eq:h30eps}) provides integer invariants in agreement
with the definition (\ref{eq:refw}).
\end{proof}

The recursive procedure explained above can be solved, such that
$h_{r,c_1}(z,\tau;J_{\varepsilon,1})$ can be directly expressed in
terms of the $H_{r'}(z,\tau;f)$ with $r'\leq r$, without computing
first the $h_{r',c_1}(z,\tau;J_{\varepsilon,1})$, and moreover giving
more compact expressions. The solution follows from  Ref. \cite{Zagier:1996} (the solution
to the recursion for vector bundles over Riemann surfaces) 
and Eq. (\ref{eq:inversestackinv}) one obtains:
\begin{eqnarray}
h_{r,c_1}(z,\tau;J_{\varepsilon,1})&=&\sum_{(r_1,c_{1,1})+\dots+(r_\ell
  ,c_{1,\ell })    =(r,c_1), \atop \mu_i\cdot J_{\varepsilon,1} \geq \mu_{i+1}\cdot J_{\varepsilon,1}}
\frac{(-1)^{m-1}}{m} w^{-\sum_{i<j} r_ir_j(\mu_j-\mu_i) \cdot K_S}
\prod_{i=1}^m H_{r_i,0}(z,\tau;f) \non \\
&=& \sum_{(r_1,a_1)+\dots+(r_\ell, a_\ell )    =(r,c_1\cdot C), \atop
  a_i\geq a_{i+1}}
\frac{(-1)^{m-1}}{m} w^{-2\sum_{i<j} r_ir_j(a_j-a_i)}
\prod_{i=1}^m H_{r_i,0}(z,\tau;f)  
\end{eqnarray}
This becomes after carrying out the sums over $a_i$ \cite{Zagier:1996}:
\begin{eqnarray}
\label{eq:solvrecursion}
h_{r,-af}(z,\tau;J_{\varepsilon,1})
&=&\sum_{(r_1,a_1)+\dots+(r_m,a_m)=(r,a)\atop a_i/r_i=a/r}
\frac{(-1)^{m-1}}{m}  \non \\
&&\prod_{i=1}^m  \left( \sum_{  r_1+\dots 
  +r_\ell=r_i} \frac{w^{2M(r_1,\dots ,r_\ell;a_i/r_i)}}{\left(1-w^{2(r_1+r_2)}\right)\dots
  \left(1-w^{2(r_{\ell-1}+r_\ell)}\right)} H_{r_1,0}(z,\tau;f) \dots
H_{r_\ell,0}(z,\tau;f)\right), 
\end{eqnarray}
where
\be
M(r_1,\dots ,r_\ell;\lambda)=\sum_{j=1}^{\ell-1}(r_j+r_{j+1})\,\{
(r_1+\dots + r_j) \lambda \} ,
\ee
with $\{ \lambda \}:=\lambda-\lfloor \lambda \rfloor$. 
 
One can verify that Eq. (\ref{eq:solvrecursion}) for $r=3$  is in
agreement with Eqs. (\ref{eq:h31eps}) and (\ref{eq:h30eps}). 
As an example we give here $h_{4,0}(z,\tau;J_{\varepsilon,1})$:
\begin{eqnarray}
h_{4,0}(z,\tau)&=&H_{4,0}(\tau,z;f)+\frac{1}{2}\frac{1+w^8}{1-w^8}\,H_{2,0}(\tau,z;f)^2+\frac{1+w^8}{1-w^8}\,H_{1,0}(\tau,z;f)\,
H_{3,0}(\tau,z;f)\non \\
&&+\frac{1-w^{16}}{(1-w^4)(1-w^6)^2}\,H_{1,0}(\tau,z;f)^2\,H_{2,0}(\tau,z;f)+\frac{1}{4}\frac{1-w^{16}}{(1-w^4)^4}H_{1,0}(\tau,z;f)^4,
\end{eqnarray}
which is to be compared with:
\begin{eqnarray}
h_{4,0}(z,\tau)&=&H_{4,0}(\tau,z;f)-\left( -\frac{w^{12}}{(1-w^8)\,(1-w^{12})^2}
  +\frac{1}{2}\frac{1+w^{24}}{(1-w^{12})\, (1-w^{24})}\right. \non \\
&&\left.\qquad  -\frac{1}{3}\frac{1}{1-w^{24}}-\frac{1}{4}\frac{1}{1-w^{16}}+\frac{1}{24}\right)h_{1,0}(z,\tau)^4\non\\
&&-\left(\frac{2(1+w^{20})}{(1-w^{16})\,(1-w^{24})}+\frac{1+w^{24}}{(1-w^{12})\,(1-w^{24})}-\frac{2}{1-w^{24}} \right. \non \\
&& \left.\qquad  -\frac{1}{1-w^{16}}+\frac{1}{2} \right) h_{1,0}(z,\tau)^2 \,h_{2,0}(z,\tau;J_{\varepsilon,1})\non
\\
&&-\left(
  \frac{2\,(w^{10}+w^{30})}{(1-w^{16})\,(1-w^{24})}+\frac{2\,w^{18}}{(1-w^{12})\,(1-w^{24})}
\right)\,h_{1,0}(z,\tau)^2 \,h_{2,f}(z,\tau;J_{\varepsilon,1})\non \\
&&-\left( -\frac{1}{1-w^{16}}+\frac{1}{2}\right)\,
h_{2,0}(z,\tau;J_{\varepsilon,1})^2 - \left( -\frac{w^8}{1-w^{16}}\right)\,h_{2,f}(z,\tau;J_{\varepsilon,1})^2\non \\
&&-\left(-\frac{2}{1-w^{24}} +1\right)\,
h_{1,0}(z,\tau)\,h_{3,0}(z,\tau;J_{\varepsilon,1})\non \\
&& -\left(-\frac{2\,(w^8+w^{16})}{1-w^{24}} \right) \,
h_{1,0}(z,\tau)\,h_{3,0}(z,\tau;J_{\varepsilon,1}).\non
\end{eqnarray}

\subsection{Wall-crossing} 
This subsection explains how to compute $h_{r,c_1}(z,\tau;J)$ for a
generic choice of polarization $J$ from the generating functions for
$J=J_{\varepsilon,1}$. The BPS invariants $\Omega(\Gamma,w;J)$ for $J$ differ in
general from those for $J=J_{\varepsilon,1}$, since sheaves might
become semi-stable or unstable by changing the polarization. The
change of the BPS invariants depends on the Hirzebruch
surface $\Sigma_\ell$ through the canonical class
$K_{\Sigma_\ell}$. Knowing how $h_{r,c_1}(z,\tau;J)$ varies in the
ample cone $C(\Sigma_1)$ is
particularly important for the computation of $h_{r,c_1}(z,\tau;\bP^2)$ since the blow-up formula is
to be applied for the polarization $J_{1,0}=\phi^*H$, where $H$ is the hyperplane class of $\bP^2$ (see the next section). The
change of the invariants can be obtained recursively from Eq. (\ref{eq:setfiltration}) after determining which
filtrations change from semi-stable to unstable or vice versa. 
 
More quantitatively one has for $J$ and $J'$ sufficiently close:
\begin{eqnarray}
\label{eq:wallcrossing}
\Delta \bar \Omega(\Gamma,w;J\to
J')&=&\sum_{{\Gamma=\Gamma_1+\dots+\Gamma_\ell ,\atop p_{J'}(\Gamma_{i})\preceq
  p_{J'}(\Gamma_{i+1}),}
  \atop p_J(\Gamma_{i})\succeq
  p_J(\Gamma_{i+1})}\frac{1}{|\mathrm{Aut}(\{\Gamma_i\};J)|} \,
w^{-\sum_{i<j} r_ir_j (\mu_j-\mu_i)\cdot K_S}\prod_{i=1}^\ell \bar
\Omega(\Gamma_i;w,J)\non\\
&&-\sum_{{\Gamma=\Gamma_1+\dots+\Gamma_\ell ,\atop p_{J'}(\Gamma_{i})\succeq
  p_{J'}(\Gamma_{i+1}),}
  \atop p_J(\Gamma_{i})\preceq
  p_J(\Gamma_{i+1})}\frac{1}{|\mathrm{Aut}(\{\Gamma_i\};J')|} \, w^{-\sum_{i<j} r_ir_j (\mu_j-\mu_i)\cdot K_S}\prod_{i=1}^\ell\bar \Omega(\Gamma_i;w,J'),
\end{eqnarray}
with $|\mathrm{Aut}(\{\Gamma_i\};J)|$ defined below
Eq. (\ref{eq:setfiltration}). Note that the invariants are evaluated
on both sides of the wall. This makes this
formula a recursive formula as it requires knowledge of
$\Omega(\Gamma_i,w;J')$, but since we are only interested in small rank
this is not a serious obstacle. A solution to the recursion is given by
Theorem 6.24 of \cite{Joyce:2004}. Other ways to determine
$\Omega(\Gamma_i,w;J')$ in terms of $\Omega(\Gamma_i,w;J)$
is using a graded Lie algebra \cite{Kontsevich:2008} or the Higgs branch analysis of
Ref. \cite{Manschot:2010qz} based on Ref. \cite{Reineke:2002}.



Since generating functions capturing wall-crossing are already described in the literature,
the explicit expressions of $h_{r,c_1}(z,\tau;J_{m,n})$ for $r=2$ and $3$, are presented here without further details. We have for $r=2$ \cite{Yoshioka:1994, Gottsche:1996}:
\begin{eqnarray}
h_{2,\beta C -\alpha f}(z,\tau; J_{m,n})&=&h_{2,\beta C -\alpha
  f}(z,\tau;J_{\varepsilon,1})+\non \\
&&\half \sum_{a,b\in
  \mathbb{Z}}\half \left(\, \sgn((2b-\beta)n-(2a-\alpha )m )- \sgn((2b-\beta)-(2a-\alpha ) \varepsilon)\,  \right)\non\\
&&\times \left(w^{-(\ell-2)(2b-\beta)-2(2a-\alpha)}-w^{(\ell-2)(2b-\beta)+2(2a-\alpha)} \right)\,q^{\frac{\ell}{4}(2b-\beta)^2+\frac{1}{2}(2b-\beta)(2a-\alpha)}\,h_{1,0}(z,\tau)^2,\non
\end{eqnarray}
and for $r=3$ \cite{Manschot:2010nc, Manschot:2010xp}:
\begin{eqnarray}
h_{3,\beta C -\alpha f}(z,\tau; J_{m,n})&=&h_{3,\beta C -\alpha
  f}(z,\tau;J_{\varepsilon,1})+\non \\
&& \sum_{a,b\in
\mathbb{Z}}\half \left(\, \sgn((3b-2\beta)n-(3a-2\alpha )m )-
\sgn((3b-2\beta)-(3a-2\alpha ) \varepsilon)\,  \right) \non \\
&&\times
\left(w^{-(\ell-2)(3b-2\beta)-2(3a-2\alpha)}-w^{(\ell-2) (3b-2\beta)+2(3a-2\alpha)}
\right)\,q^{\frac{\ell}{12}(3b-2\beta)^2+\frac{1}{6}(3b-2\beta)(3a-2\alpha)}\non\\
&&\times h_{2,bC-af}(z,\tau;\Sigma_\ell, J_{|3b-2\beta|,|3a-2\alpha|})\,h_{1,0}(z,\tau).\non
\end{eqnarray}

\section{BPS invariants of $\bP^2$}
\label{sec:projplane}

The Hirzebruch surface $\Sigma_1$ can be obtained as a blow-up $\phi:\Sigma_1\to
\bP^2$ of the projective plane $\bP^2$. Interestingly, we can compute the BPS invariants of $\bP^2$
from those of $\Sigma_1$ from the blow-up formula.
This formula is a remarkable result which states that the ratio of generating functions of
BPS invariants of a surface $S$ and its blow-up $\tilde S$ is a (theta) function
independent of $S$ or $J$ \cite{Yoshioka:1996,
  Gottsche:1998, Li:1999}. The underlying reason for this relation is
that every semi-stable sheaf on $\tilde S$ can be obtained from one on $S$
by an elementary transformation along the exceptional divisor of the blow-up. 

Two subtle issues of the blow-up formula are (Proposition 3.4 of \cite{Yoshioka:1996}):
\begin{itemize}
\item[-] the stability condition is $\mu$-stability rather than Gieseker stability,
\item[-] it involves the virtual Poincar\'e functions $\CI(\Gamma,w;J)$ of the moduli stack.
\end{itemize} 

To take these two issues into account let $\bar
\Omega^\mu(\Gamma,w;J)$ be the invariant enumerating $\mu$-semi-stable 
sheaves which is obtained from $\bar \Omega^\mu(\Gamma,w;J)$ by
addition of the Gieseker unstable sheaves which are $\mu$-semi-stable
using Eq. (\ref{eq:setfiltration}). Moreover, let
$\CI^\mu(\Gamma,w;J)$ be the corresponding virtual Poincar\'e
function with corresponding generating function $H^\mu_{r,c_1}(z,\tau;\tilde S,J)$. 

The blow-up formula now reads \cite{Yoshioka:1996, Gottsche:1998, Li:1999}: 
\begin{proposition}
\label{prop:blowup}
Let $S$ be a smooth projective surface and $\phi: \tilde S \to S$ the
blow-up at a non-singular point, with $C_\mathrm{e}$ the exceptional
divisor of $\phi$. The generating functions $H^\mu_{r,c_1}(z,\tau;S,J)$ and
$H^\mu_{r,c_1}(z,\tau;\tilde S,J)$ are related by the ``blow-up
formula'':
\be
\label{eq:blowup}
H^\mu_{r,\phi^* c_1-kC_\mathrm{e}}(z,\tau;\tilde S,\phi^*J)=B_{r,k}(z,\tau)\, H^\mu_{r,c_1}(z,\tau;S,J),
\ee
with
\be
B_{r,k}(z,\tau)=\frac{1}{\eta(\tau)^r}\sum_{\sum_{i=1}^ra_i=0 \atop a_i \in \mathbb{Z}+\frac{k}{r}} q^{-\sum_{i<j}a_ia_j}w^{\sum_{i<j}a_i-a_j}.\non
\ee 
The blow-up formula for generating functions of Hodge numbers is
identical except with the replacement of $z$ by  $\half(u+v)$ in
$B_{r,k}(z,\tau)$.
\end{proposition}
\noi The two relevant cases for this article are $r=2,3$:
\be
B_{2,k}(z,\tau)=\frac{\sum_{n\in \mathbb{Z}+k/2} q^{n^2}w^n}{\eta(\tau)^2},\qquad B_{3,k}(z,\tau)=\frac{\sum_{m,n \in \mathbb{Z}+k/3} q^{m^2+n^2+mn}w^{4m+2n}}{\eta(\tau)^3}.
\ee
Note that $B_{r,k}(z,\tau)$ does not depend on $S$ or $J$.

The computation of $h_{r,c_1}(z,\tau;\bP^2)$ from $h_{r,\phi^* c_1-kC}(z,\tau;\Sigma_1)$ in general involves the
following three steps: 
\begin{enumerate}
\item  Compute $h^\mu_{r,\phi^* c_1-kC}(z,\tau; J_{1,0})$ by adding to
  $h_{r,\phi^* c_1-kC}(z,\tau; J_{1,\varepsilon})$ terms due to
  sheaves on $\Sigma_1$ which are not Gieseker stable for
  $J_{1,\varepsilon}$, but   $\mu$-semistable for $\phi^* H=J_{1,0}$,
  and consequently compute  $H^\mu_{r,\phi^* c_1-kC}(z,\tau; J_{1,0})$
  by adding the terms prescribed by
  Eq. (\ref{eq:stackinvariant}). The generating functions and the
  factorial factors in   Eq. (\ref{eq:stackinvariant}) combine these two steps very naturally
  into one.
\item Divide by $B_{r,k}(z,\tau)$ to obtain $H^\mu_{r,c_1}(z,\tau; \bP^2)$.
\item Determine $h_{r,c_1}(z,\tau; \bP^2)$ from $H^\mu_{r,c_1}(z,\tau; \bP^2)$ by reversing step (1).
\end{enumerate}

For $c_1=\beta C+f$, $\beta=0$ or $1$, and $J=J_{1,0}$,
$\mu$-stability is equivalent to Gieseker stability, and therefore
steps 1) and 3) become trivial. For example, one can compute
$h_{3,H}(z,\tau;\bP^2)$ starting from $h_{3,C+f}(z,\tau;J_{1,0})$ as
was done in Ref. \cite{Manschot:2010nc}, or from from $h_{3,f}(z,\tau;J_{1,0})$
which requires Conjecture \ref{conj:restrictfibre} and Eq. (\ref{eq:setfiltration}). One can verify that the first terms of both $q$-expansions of
$h_{3,H}(z,\tau;\bP^2)$ are equal, which is in agreement with Proposition
\ref{prop:blowup}. A proof of the equality of these expressions for
$h_{3,H}(z,\tau;\bP^2)$ would imply a proof of Conjecture
\ref{conj:restrictfibre} for $(r,c_1)=(3,f)$ since 
$h_{3,f}(z,\tau;J_{\varepsilon,1})$ is related to $h_{3,C+f}(z,\tau;J_{1,\varepsilon})$ by the blow-up formula and
wall-crossing.  


When $\mu$- and Gieseker stability are not equivalent, steps 1) and 3)
are not trivial. We will first explain them for $r=2$ following \cite{Yoshioka:1995}.  One obtains:
\begin{proposition}
\label{eq:h20P2} 
\begin{eqnarray}
h_{2,0}(z,\tau;\bP^2)=\frac{1}{B_{2,1}(z,\tau)} \left[
  h_{2,C}(z,\tau;J_{1,\varepsilon}) +\sum_{b<0 \atop b=-1 \mod 2} w^b q^{\frac{1}{4}b^2}h_{1,0}(z,\tau)^2\,\right]-\frac{1}{2}h_{1,0}(z,\tau;\bP^2)^2.\non
\end{eqnarray}
\end{proposition}
\begin{proof}
The only extended HN-filtrations which are Gieseker unstable for
$J=J_{1,\varepsilon}$ and $\mu$-semi-stable for $J=J_{1,0}$ have $\ell=\ell_\mu=2$. 
For the parametrization $c_1(E_2)=bC-af$, the set of sheaves which is unstable for
$J_{1,\varepsilon}$ but $\mu$-semistable for $J_{1,0}$ corresponds to
$b<0$ and $a=0$. This gives the second term inside the brackets. 
Consequently, step (2) divides by $B_{2,1}(z,\tau)$, and step (3)
subtracts the $\mu$-semi-stable sheaves which are not Gieseker
semi-stable with $\ell=2$ and $\ell_\mu=1$. 
\end{proof}

Alternatively, one can compute $h_{2,0}(z,\tau;\bP^2)$ starting from
$h_{2,0}(z,\tau;J_{1,\varepsilon})$. In that case the term due to step
(1) in the brackets is $\left( \sum_{b<0 \atop b=0 \mod
    2}w^bq^{\frac{1}{4}b^2}+\frac{1}{2} \right)\,h_{1,0}(z,\tau)^2$,
and one divides by $B_{2,0}(z,\tau)$. Addition of
$\frac{1}{2}h_{1,0}(z,\tau;\bP^2)$ provides the expected integer
invariants, in agreement with \cite{Yoshioka:1995}.
 Accidentily, the terms due to step (1) and step (3) can simply be incorporated by replacing
$J_{1,\varepsilon}$ by $J_{1,0}$ in $h_{2,\beta
  C}(z,\tau;J_{1,\varepsilon})$,  and can be written in terms of
the Lerch sum \cite{Bringmann:2010sd}.

The remainder of this section discusses $r=3$. In terms of $h_{3,C}(z,\tau;J_{1,\varepsilon})$, $h_{3,0}(z,\tau;\bP^2)$ is given by:
\begin{proposition}
\begin{eqnarray}
h_{3,0}(z,\tau;\bP^2)&=&\frac{1}{B_{3,1}(z,\tau)} \left[
  h_{3,C}(z,\tau;J_{1,\varepsilon}) +\left( \sum_{b< 0 \atop b=-2,-4 \mod 6 }
  w^{b}q^{\frac{1}{12}b^2}\right)\, h_{1,0}(z,\tau) \,
h_{2,0}(z,\tau;J_{1,\varepsilon}) \right. \non \\
&&+ \left( \sum_{b<0 \atop b=-1,-5 \mod 6 }
  w^{b}q^{\frac{1}{12}b^2}\right)\, h_{1,0}(z,\tau) \,
h_{2,C}(z,\tau;J_{1,\varepsilon}) \\
&&+\left.\left(\sum_{k_1,k_2<0\atop k_2=k_1+1 \mod
      3}w^{2(k_1+k_2)}q^{\frac{1}{3}(k_1^2+k_2^2+k_1k_2)}+\frac{1}{2}  \sum_{k<0, \atop k=-1,-2 \mod 3}
    w^{2k}q^{\frac{1}{3}k^2} \right)\,h_{1,0}(z,\tau)^3\right] \non \\
&&-\frac{1}{6} h_{1,0}(z,\tau;\bP^2)^3-\frac{2}{2}\,h_{1,0}(z,\tau;\bP^2)\,h_{2,0}(z,\tau;\bP^2).\non
\end{eqnarray}

\end{proposition}
The desired integer invariants are obtained from
$h_{3,0}(z,\tau;\bP^2)$ after subtraction of $\frac{1}{3}h_{1,0}(3z,3\tau;\bP^2)=\frac{1}{3}\frac{i}{\theta(6z,3\tau)}$. The first non-vanishing coefficients
are presented in Table \ref{tab:betti30}. They are in agreement with
the expected dimension of $\CM(\Gamma)$ (\ref{eq:dim}). 

\begin{table}[h!]
\begin{tabular}{lrrrrrrrrrrrrrrrrr}
$c_2$ & $b_0$ & $b_2$ & $b_4$ & $b_6$ & $b_8$ & $b_{10}$ & $b_{12}$ & $b_{14}$
& $b_{16}$ & $b_{18}$ & $b_{20}$ & $b_{22}$ & $b_{24}$ & $b_{26}$ & $b_{28}$  &  $\chi$ \\
\hline
3 & 1 & 1 & 2 & 2 & 2 & 2 & & & & & & & & & & 18  \\
4 & 1 & 2 & 5 & 9 & 15 & 19 & 22 & 23 & 24 & & & & & & & 216 \\
5 & 1 & 2 & 6 & 12 & 25 & 43 & 70 & 98 & 125 & 142 & 154 & 156 & & & &1512 \\
6 & 1 & 2 & 6 & 13 & 28 & 53 & 99 & 165 & 264 & 383 & 515 & 631 &
723 & 774 & 795 & 8109 
\end{tabular}
\caption{The Betti numbers $b_n$ (with $n\leq
  \dim_\mathbb{C} \mathcal{M}$) and the Euler number $\chi$ of the moduli spaces of semi-stable sheaves
  on $\bP^2$ with $r=3$, $c_1=0$, and $3\leq c_2\leq 6$.}  
\label{tab:betti30}
\end{table}

\begin{proof}
The terms added to $h_{3,C}(z,\tau;J_{1,\varepsilon})$ in the brackets
are due to step (1). The last term on the first line and the term on
the second line are due to filtrations with $\ell=\ell_\mu=2$. If one
chooses $c_1(E_2)=bC-af$ as for $r=2$, the set of sheaves which are
unstable for $J_{1,\varepsilon}$ but 
$\mu$-semistable for $J_{1,0}$ corresponds to $b<0$ and $a=0$. Similarly, the
first term in parentheses on the third line is due to $\ell=\ell_\mu=3$, and the
second term due to $\ell=3$ and $\ell_\mu=2$. The sum of the terms in
the bracket is $H^\mu_{3,C}(z,\tau;J_{1,\varepsilon})$, and is divided
by $B_{3,1}(z,\tau)$ following step (2). Finally, step (3)
corresponds to the last line. 
\end{proof} 


%
%
%

As a consistency check, $h_{3,0}(z,\tau;\bP^2)$ can also be computed from
$h_{3,0}(z,\tau; J_{1,\varepsilon})$. Then the terms due to step (1)
are for $\ell=2$:
\begin{eqnarray} 
\label{eq:contbd4}
&&\left(\frac{2}{2}+ 2\sum_{b<0 \atop b=0 \mod 6 }
  w^{b}q^{\frac{1}{12}b^2}\right)\, h_{1,0}(z,\tau) \,
h_{2,0}(z,\tau;J_{1,\varepsilon})\\
&&\qquad + \left(2\sum_{b<0 \atop b=-3 \mod 6 }
  w^{b}q^{\frac{1}{12}b^2}\right)\, h_{1,0}(z,\tau) \,
h_{2,C}(z,\tau;J_{1,\varepsilon}),\non 
\end{eqnarray}  
and for $\ell=3$:
\be
\left(\sum_{k_1,k_2<0\atop k_1=k_2 \mod 3}w^{2(k_1+k_2)}q^{\frac{1}{3}(k_1^2+k_2^2+k_1k_2)}+ \frac{2}{2} \sum_{k<0 \atop k=0 \mod 3} w^{2k}q^{\frac{1}{3}k^2}+\frac{1}{6}\right)\,h_{1,0}(z,\tau)^3.
\ee

We conclude by briefly comparing the BPS invariants computed above to
the results obtained by Refs. \cite{Kool:2009,
  weist:2009} for Euler numbers of moduli spaces using toric
localization of the moduli spaces. Ref. \cite{weist:2009} computed
such Euler numbers for $\mu$-stable vector
bundles \cite{weist:2009} with rank $r\leq 3$ on $\bP^2$, whereas
Ref. \cite{Kool:2009}  computed such Euler numbers for $\mu$-stable
torsion free sheaves with rank $r\leq 3$ on various smooth toric
surfaces. If $\gcd(r,c_1)=1$ and for a generic choice of polarization,
the moduli space of $\mu$-stable sheaves is isomorphic to the
moduli space of Gieseker semi-stable sheaves.  Otherwise, the
moduli space of $\mu$-stable sheaves is a smooth open subset 
of the moduli space of Gieseker semi-stable sheaves. The difference
between generating functions of Euler numbers for vector bundles and
torsion free sheaves is an overall factor $\eta(\tau)^{r\chi(S)}$.

For Chern classes such that $\mu$-stability is equivalent to Gieseker semi-stability,
agreement of Refs. \cite{Kool:2009, weist:2009} with the
techniques described in this paper is expected. This is indeed established in
Refs. \cite{Manschot:2010nc, Manschot:2011dj}. In particular, Eq. (4.5) and Table 1 of Ref.
\cite{Manschot:2010nc} agree with Corollary 4.10 in Ref. \cite{weist:2009} and
Corollary 4.9 in Ref. \cite{Kool:2009}. If $\gcd(r,c_1)>1$ strictly Gieseker semi-stable
sheaves can occur and therefore agreement of $h_{r,0}(z,\tau;\bP^2)$ with Refs. \cite{Kool:2009,
  weist:2009} is not expected. Indeed the numbers in Table 1 above
appear to be different from the Euler numbers computed by Theorem 4.14
of Ref. \cite{weist:2009} and Corollary 4.9 of Ref. \cite{Kool:2009}. It would be interesting to precisely
understand the difference between the Euler numbers of the $\mu$-stable loci and the
BPS invariants computed above.


\providecommand{\href}[2]{#2}\begingroup\raggedright\endgroup


\begin{thebibliography}{1}

\bibitem{Atiyah:1982fa}
  M.~F.~Atiyah, R.~Bott,
  {\it The Yang-Mills equations over Riemann surfaces,}
  Phil.\ Trans.\ Roy.\ Soc.\ Lond.\  {\bf A308 } (1982)  523-615.

\bibitem{Beauville:1992}
A.~Beauville.
{\it Sur la cohomologie de certaines espaces de modules de fibr\'es vectoriels},
 {\em Geometry and Analysis}, (Bombay, 1992), Tata Inst. Fund. Res. 37-40, (1995),
 arXiv:math.AG/9202024.

\bibitem{Bringmann:2010sd}
  K.~Bringmann and J.~Manschot,
  {\it From sheaves on $\mathbb{P}^2$ to a generalization of the Rademacher expansion},
  Am. J. of Math., arXiv:1006.0915 [math.NT].

\bibitem{Chuang:2013}
 W.-Y~Chuang, D.-E.~Diaconescu, J.~Manschot, G.~W.~Moore and Y.~Soibelman,
  {\it Geometric engineering of (framed) BPS states}, arXiv:1301.3065 [hep-th].
 
\bibitem{Desale:1975}
U.~V.~Desale and S.~Ramanan, {\it Poincar\'e polynomials of the variety
of stable bundles}, Math.\ Ann. {\bf 216} (1975) 233-244.

\bibitem{Diaconescu:2007bf}
  E.~Diaconescu and G.~W.~Moore,
  {\it Crossing the wall: Branes versus bundles,}
  Adv.\ Theor.\ Math.\ Phys.\  {\bf 14} (2010),
  arXiv:0706.3193 [hep-th].

\bibitem{Earl:2000}
  R.~Earl and F.~Kirwan,
  {\it The Hodge Numbers of the Moduli Spaces of Vector Bundles over a
    Riemann Surface,} Q. J. Math. {\bf 51} (2000)  465-484,
[arXiv:math/0012260]

  
\bibitem{Friedman:1998}
  R.~Friedman,   ``Algebraic Surfaces and Holomorphic Vector
  Bundles,'' Springer-Verlag (1998).

\bibitem{Gaiotto:2010be}
  D.~Gaiotto, G.~W.~Moore and A.~Neitzke, 
  {\it Framed BPS States},
  arXiv:1006.0146 [hep-th].

\bibitem{Gottsche:1990}
L.~G\"ottsche, {\it The Betti numbers of the Hilbert scheme of points
  on a smooth projective surface}, Math.\ Ann. {\bf 286} (1990) 193.

\bibitem{Gottsche:1996}
  L.~G\"ottsche, D.~Zagier,
  {\it Jacobi forms and the structure of Donaldson invariants for 4-manifolds with $b_+=1$,}
Selecta \ Math. New Ser. {\bf 4} (1998) 69.   
  [arXiv:alg-geom/9612020].

\bibitem{Gottsche:1998}
  L.~G\"ottsche,
  {\it Theta functions and Hodge numbers of moduli spaces of sheaves on rational surfaces,}
  Comm.\ Math.\ Physics \ {\bf 206} (1999) 105
  [arXiv:math.AG/9808007].
 
\bibitem{Harder:1975}
G.~Harder, M.~S.~Narasimhan, {\it On the cohomology groups of moduli
  spaces of vector bundles on curves,}  Math. Ann. {\bf 212} (1975) 215-248.
 
\bibitem{Huybrechts:1996}
  D.~Huybrechts and M.~Lehn,
  ``The geometry of moduli spaces of sheaves,'' Cambridge University
  Press (1996).

\bibitem{Joyce:2004}
  D.~Joyce, {\it Configurations in Abelian categories. IV. Invariants
    and changing stability conditions},
  [arXiv:math.AG/0410268].

\bibitem{Joyce:2005}
  D.~Joyce, {\it Motivic invariants of Artin stacks and 'stack functions'},
  [arXiv:math.AG/0509722]. 

\bibitem{Joyce:2008}
  D.~Joyce and Y.~Song, {\it A theory of generalized Donaldson-Thomas invariants},
  arXiv:0810.5645 [math.AG].

\bibitem{Katz:1996fh}
  S.~H.~Katz, A.~Klemm and C.~Vafa, 
{\it Geometric engineering of quantum field theories,}
  Nucl.\ Phys.\ B {\bf 497} (1997) 173
  [hep-th/9609239].

\bibitem{Kim:2011sc}
  H.~Kim, J.~Park, Z.~Wang and P.~Yi,
  {\it Ab Initio Wall-Crossing},
  JHEP {\bf 1109} (2011) 079
  [arXiv:1107.0723 [hep-th]].

\bibitem{Kirwan:1986}
  F.~Kirwan,
  {\it On the homology of compactifications of moduli spaces of vector
    bundles over Riemann surfaces,}
Proc. London Math. Soc. {\bf 53} (1986) 237-266; {\it Corrigendum}, Proc. London Math. Soc. {\bf 65} (1992) 474.
 
\bibitem{Kontsevich:2008}
  M.~Kontsevich and Y.~Soibelman,
  {\it Stability structures, motivic Donaldson-Thomas invariants and
  cluster transformations}, [arXiv:0811.2435 [math.AG]].

\bibitem{Kool:2009}
  M.~Kool,
  {\it Euler charactertistics of moduli spaces of torsion free sheaves on toric surfaces,}
  arXiv:0906.3393 [math.AG]. 

\bibitem{Li:1999}
  W.-P.~Li and Z. Qin,
  {\it On blowup formulae for the $S$-duality conjecture of Vafa and Witten,}
  Invent. Math. \ {\bf 136} (1999) 451-482
  [arXiv:math.AG/9808007].

\bibitem{Manschot:2009ia}
  J.~Manschot,
  {\it Stability and duality in $\CN=2$ supergravity}, Commun.\ Math.\ Phys. {\bf 299} (2010) 651-676, 
  arXiv:0906.1767 [hep-th].

\bibitem{Manschot:2010xp}
  J.~Manschot,
  {\it Wall-crossing of D4-branes using flow trees}, Adv. Theor. Math.
  Phys. {\bf 15} (2011) 1-42, arXiv:1003.1570 [hep-th].

\bibitem{Manschot:2010nc}
  J.~Manschot,
  {\it The Betti numbers of the moduli space of stable sheaves of rank
    3 on $\bP^2$}, Lett. in Math. Phys. {\bf 98} (2011) 65-78,
    [arXiv:1009.1775 [math-ph]].

\bibitem{Manschot:2010qz}
  J.~Manschot, B.~Pioline, A.~Sen,
  {\it Wall Crossing from Boltzmann Black Hole Halos,}
  JHEP {\bf 1107 } (2011)  059.
  [arXiv:1011.1258 [hep-th]].

\bibitem{Manschot:2011dj}
  J.~Manschot,
  {\it BPS invariants of $\CN=4$ gauge theory on a surface},
  Comm. Number Th. and Phys. {\bf 2} (2012)
  [arXiv:1103.0012 [math-ph]].

\bibitem{Nakajima:2010}
H.~Nakajima and K.~Yoshioka,
{\it Instanton counting and Donaldson invariants},
Sugaku Expositions {\bf 23} (2010) 2.

\bibitem{Reineke:2002}
M.~Reineke,
{\it The Harder-Narasimhan system in quantum groups and cohomology of
  quiver moduli}, Invent. Math. {\bf 152} (2003) 349.
 
\bibitem{Vafa:1994tf}
  C.~Vafa and E.~Witten,
  {\it A strong coupling test of S duality, }
  Nucl.\ Phys.\  B {\bf 431} (1994) 3
  [arXiv:hep-th/9408074].

\bibitem{weist:2009}
T.~Weist, {\it Torus fixed points of moduli spaces of stable bundles
  of rank three}, J. of Pure and Applied Algebra {\bf 215} (2011), arXiv:0903.0732 [math. AG].

\bibitem{Yoshioka:1994}
K.~Yoshioka, {\it The Betti numbers of the moduli space of stable sheaves of
  rank 2 on $\mathbb{P}^2$},  J. reine. angew. Math. {\bf 453} (1994)
  193--220.

\bibitem{Yoshioka:1995}
K.~Yoshioka, {\it The Betti numbers of the moduli space of stable
  sheaves of rank 2 on a ruled surface},  Math. Ann. {\bf 302} (1995)
  519--540.

\bibitem{Yoshioka:1996}
K.~Yoshioka, {\it The chamber structure of polarizations and the
  moduli of stable sheaves on a ruled surface},  Int. J. of Math. {\bf 7} (1996)
  411--431 [arXiv:alg-geom/9409008].

\bibitem{Yoshioka:1998ti}
  K.~Yoshioka,
  {\it Euler characteristics of SU(2) instanton moduli spaces on rational
  elliptic surfaces,}
  Commun.\ Math.\ Phys.\  {\bf 205} (1999) 501
  [arXiv:math/9805003].

\bibitem{Zagier:1996}
  D.~Zagier, {\it Elementary aspects of the Verlinde formula and of the
    Harder-Narasimhan-Atiyah-Bott formula,} Proceedings of the
  Hirzebruch 65 Conference on Algebraic Geometry (Ramat Gan, 1993),
  445-462, Israel Math. Conf. Proc., 9, Bar-Ilan Univ., Ramat Gan, 1996. 



\end{thebibliography}
\end{document}